\documentclass[runningheads]{llncs}
\usepackage{listings}
\usepackage[backend=biber,
            style=numeric,
            sortlocale=en_EN,
            natbib=true,
            url=false, 
            doi=true,
            eprint=false]{biblatex} 

\usepackage{stmaryrd}
\usepackage{amsmath}
\usepackage{amssymb}
\usepackage{amsfonts}
\usepackage{algorithm}
\usepackage[noend]{algpseudocode}

\makeatletter
\def\BState{\State\hskip-\ALG@thistlm}
\makeatother

\usepackage{tikz}

\newcommand{\Rp}{\mathbb{R}_+}
\newcommand{\N}{\mathbb{N}}

\newcommand{\pr}{\mathrm{pr}}
\newcommand{\R}{\mathcal{R}}

\newcommand{\Po}{\mathcal{P}}

\newcommand{\TES}[1]{\mathit{TES}(#1)}

\title{Runtime Composition Of Systems of Interacting Cyber-Physical Components}
\author{Benjamin Lion\inst{1}, Farhad Arbab\inst{1,2}, Carolyn Talcott\inst{3}}
 \institute{
            CWI, Amsterdam, The Netherlands \\ 
            \email{lion@cwi.nl} \and 
     Leiden University, Leiden, The Netherlands \\ 
     \email{arbab@cwi.nl}  \and 
    SRI International, CA, USA \\ 
    \email{carolyn.talcott@gmail.com}}
\begin{document}

\maketitle

\begin{abstract}
We introduce a transition system based specification of cyber-physical systems whose semantics is compositional with respect to a family of algebraic products.
We give sufficient conditions for execution of a product to be correctly implemented by a lazy expansion of the product construction.
The transition system algebra is implemented in the Maude rewriting logic system, and we report a simple case study illustrating compositional specification.
\end{abstract}

\section{Introduction}
We model a cyber-physical system as a pair consisting of an interface and a behavior. 
We call such pairs components.
The interface is a set of events and the behavior is a set of infinite sequences of observations called Timed-Event Stream (TES), where an observation is a pair of a set of events from the interface and a time stamp. 

In previous work,
we define a family of parametrized binary products to compose components. 
The parameters of a product consist of a composability relation on TESs, which says which pairs of TESs can compose, and a composition function, which says how two TESs compose to form a new TES.
For instance, the synchronous product of two components with shared events composes pairs of TESs, one from each component behavior, such that shared events occur at the same time. The resulting TES interleaves observations from a pair of composable TESs.
In order to bridge the gap between a denotational model and an operational model, we give
a co-inductive definition of composability relations on TESs as a lift of composability relations on observations. Such locally defined relations provide a step-wise mechanism to check whether two TESs are composable. 

In this paper, we give an operational mechanism to specify components by transition systems labeled with observations, called TES transition systems. Different TES transition systems may denote the same component, as a component is oblivious to internal non-determinism of the machinery that manifests its behavior.

As for components, we introduce a family of parametrized algebraic products on TES transition systems. The parameter here is a composability relation on observations, and each transition in the product is the result of the composition of a pair of transitions with composable labels. We show that the TES transition system component semantics is compositional with respect to such products, i.e., the component resulting from the product of two TES transition systems is equal to the product of the components resulting from each TES transition system. On components, the composability relation is co-inductively lifted from observations to TESs, and the composition function is set union on observations and interleaving on streams.

Because the composability relation on observations is a step-wise operation, it may lead to deadlock states in the product TES transition system, i.e., states with no outgoing transitions.
We call two TES transition systems \emph{compatible} with respect to a composability relation on observations if, for every reachable pair of states, there is at least one pair of transitions whose pair of labels is composable.
We give some sufficient conditions for TES transition systems to be compatible, and show that if two TES transitions system are compatible, then their product can be done lazily, i.e., step by step at runtime.

To carry out experiments, we have implemented TES transition systems in the Maude rewriting logic system.
We use this implementation to analyze properties of the behavior of a set of self-sorting robots.
In this example, each robot has a unique integer identifier, its own battery, and interacts with others via a shared grid. 
Robots, batteries, and grids are all modeled as TES transition systems.
We add a set of \emph{swap} protocols each of which coordinates a pair of robots to swap their positions when a robot faces another robot with a lower identifier and a higher $x$-axis coordinate. 
We analyze the resulting system by querying if a robot can reach its sorted configuration on the grid. 
The interaction among all TES transition systems is captured by a suitable composability relation on observations. 

Our work has the following benefits.
First, performing lazy composition keeps the representation of an interacting system small. 
Second, step-wise runtime composition renders our runtime framework modular, where run-time replacement of individual components becomes possible (as long as the update complies with some rules). 
Finally, the runtime framework more closely matches the architecture of a distributed framework, 
where entities are physically separated and no party may have access to the whole description of the entire system.

\newcommand{\Co}[1]{\mathit{Co(#1)}}
\newcommand{\E}[1]{\mathit{E(#1)}}
\newcommand{\Fin}[1]{\mathcal{L}^{\mathrm{fin}}(#1)}
\newcommand{\Inf}[1]{\mathcal{L}^{\mathrm{inf}}(#1)}
\newcommand{\Lan}[1]{\mathcal{L}(#1)}
\newcommand{\FG}[1]{\mathit{FG}(#1)}
\newcommand{\Rel}{\mathcal{R}}
\newcommand{\ssync}{\mathit{sync}}
\newcommand{\tes}{\textit{TES}}
\newcommand{\sorted}{\mathit{sorted}}
\newcommand{\ind}{\mathit{ind}}
\renewcommand{\E}{\mathbb{E}}

\section{Components in interaction}

In~\cite{DBLP:journals/corr/abs-2110-02214}, we give a unified semantic model to capture cyber and physical aspects of processes as components and characterize their various types of interactions as user-defined products in an algebraic framework.
Moreover, we show some general conditions for products on components to be associative, commutative, and idempotent.
In this section, we recall the basic definitions of a component and product from~\cite{DBLP:journals/corr/abs-2110-02214}, and introduce in Section 2.2 some instances of product that suit our example in this paper.

\paragraph{Notations.}
Given $\sigma: \N \to \Sigma$, let $\sigma[n] \in \Sigma^n$ be the finite prefix of size $n$ of $\sigma$ and let $\sim_n$ be an equivalence relation on $(\N\to\Sigma)\times(\N\to\Sigma)$ such that $\sigma \sim_n \tau$ if and only if $\sigma[n] = \tau[n]$.
Let $\FG{L}$ be the set of \emph{left factors} of a set $L \subseteq \Sigma^\omega$, defined as $\FG{L} = \{ \sigma[n] \mid n \in \N,\ \sigma \in L\}$. We write $\sigma^{(n)}$ for the $n$-th derivative of $\sigma$, i.e., the stream such that $\sigma^{(n)}(i) = \sigma(n+i)$ for all $i \in \N$.

A timed-event stream $\sigma \in \TES{E}$ over a set of events $E$ is an infinite sequence of \emph{observations}, where an observation $\sigma(i) = (O, t)$ consists of a pair of a subset of events in $O \subseteq E$, called \emph{observable}, and a positive real number $t \in \mathbb{R}_+$ as time stamp.
A timed-event stream (TES) has the additional properties that consecutive time stamps are increasing and non-Zeno, i.e., for any TES $\sigma$ and any time $t \in \mathbb{R}$, there exists an element $\sigma(i) = (O_i, t_i)$ in the sequence such that $t< t_i$.
We use $\sigma'$ to denote the derivative of the stream $\sigma$, such that $\sigma'(i) = \sigma(i+1)$ for all $i \in \N$.

We recall the greatest post fixed point of a monotone operator, that we later use as a definition scheme and as a proof principle.
Let $X$ be any set and let $\Po(X) = \{ V \mid V \subseteq X\}$ be the set of all its subsets.
If $\Psi : \Po(X) \rightarrow \Po(X)$ is a monotone operator, that is, $R \subseteq S$ implies $\Psi(R) \subseteq \Psi(S)$ for all $R \subseteq X$ and $S \subseteq X$, then $\Psi$ has a greatest fixed point $P = \Psi(P)$ satisfying:
$$
P = \bigcup \{ R \mid R \subseteq \Psi(R)\}
$$
This equality can be used as a proof principle: in order to prove that $R \subseteq P$, for any $R \subseteq X$, it suffices to show that $R$ is a post-fixed point of $\Psi$, that is, $R \subseteq \Psi(R)$.

\subsection{Components}
\label{section:components}
Components are the parts that forms cyber-physical systems. Therefore, a component uniformly models both cyber and physical aspects through sequence of observables. 
 \begin{definition}[Component]
     \label{def:component}
     A \emph{component} $C = (E,L)$ is a pair of an \emph{interface} $E$, and a \emph{behavior} $L \subseteq \TES{E}$.
\hfill$\triangle$
 \end{definition}
 The product of two components defines a new component whose behavior is function of the behavior of both components, under some constraints. We call a \emph{composability relation} $R(E_1, E_2) \subseteq \TES{E_1} \times \TES{E_2}$ such constraint, and say that $(\sigma_1, \sigma_2)$ are two \emph{composable} TESs if $(\sigma_1, \sigma_2) \in R(E_1, E_2)$.
 Furthermore, two TESs may then compose to form a new TES. We call a \emph{composition function} $\oplus : \TES{E_1} \times \TES{E_2} \rightarrow \TES{E_1 \cup E_2}$ such function, and write $\sigma_1 \oplus \sigma_2$ for the composition of $\sigma_1$ and $\sigma_2$.

     Let $C_1 = (E_1, L_1)$ and $C_2 = (E_2, L_2)$ be two components. 
     Let $R(E_1, E_2)$ be a composability relation and $\oplus$ be a composition function.
 \begin{definition}[Product]
     \label{def:prod}
     The \emph{product} of $C_1$ and $C_2$ under $(R,\oplus)$ is the component $C = C_1 \times_{(R,\oplus)} C_2 = (E_1 \cup E_2, L)$  where
     \[
         \{ \sigma_1 \oplus \sigma_2 \mid \sigma_1 \in L_1, \sigma_2 \in L_2,\ (\sigma_1, \sigma_2) \in R(E_1, E_2)\}
     \]
\hfill$\triangle$
 \end{definition}

While the behaviors of a component are streams, it is natural to consider termination of a component. 
We express a terminating behavior of component $C= (E,L)$ as an element $\sigma \in L$ such that there exists $n \in \N$ with $\sigma^{(n)} \in \TES{\emptyset}$. In other words, a terminating behavior $\sigma$ is such that, starting from the $n$-th observation, all next observations are empty. 

Given a component $C$, we define $C^*$ to be the component that may terminate after every sequence of observables. 
Formally, $C^*$ is the component whose behavior is the prefix closure of $C$, i.e., the component $C^* = (E, L^*)$, where 
\[
L^* = L \cup\{ \tau \mid  \exists n \in \N. \exists \sigma \in L.\ \tau \sim_n \sigma,\ \tau^{(n)} \in \TES{\emptyset} \}
\]

In~\cite{DBLP:journals/corr/abs-2110-02214}, we give a co-inductive definition for some $R$ and $\oplus$ given a composability relation on observations, and a composition function on observations. 
We use $\kappa(E_1, E_2) \subseteq (\Po(E_1) \times \Rp) \times (\Po(E_2) \times \Rp)$ to range over composability relation on observations.
    Let $\kappa$ be a parametric composability relation on observations, and let $\Phi_\kappa(E_1, E_2)$
    be such that, for any $\Rel \subseteq \TES{E_1} \times \TES{E_2}$:
 \[
\begin{array}{rl}
    \Phi_\kappa(E_1,E_2)(\Rel) = \{(\tau_1, \tau_2) \mid & (\tau_1(0), \tau_2(0)) \in \kappa(E_1,E_2) \land   \\
                                                         & (\pr_2(\tau_1)(0) =t_1 \land  \pr_2(\tau_2)(0) =t_2) \land \\
                                                         &(t_1 < t_2 \land (\tau_1',\tau_2) \in \Rel \lor t_2 < t_1 \land (\tau_1,\tau_2') \in \Rel \lor \\
                                                         &\quad t_2 = t_1 \land (\tau_1',\tau_2') \in \Rel)\}
\end{array}
\]
The \emph{lifting} of $\kappa$ on \tes s, written $[\kappa]$, is the parametrized relation obtained by taking the greatest post fixed point of the function $\Phi_\kappa(E_1, E_2)$ for arbitrary pair $E_1, E_2 \subseteq \E$, i.e., the relation $[\kappa](E_1, E_2) = \bigcup_{\Rel \subseteq \TES{E_1}{} \times \TES{E_2}{}} \{ \Rel \mid \Rel \subseteq \Phi_\kappa(E_1, E_2)(\Rel)\}$.

    Then, let $\sqcap \subseteq \Po(\E)^2$ be a relation on observables. 
        We say that two observations are synchronous under $\sqcap$ if, intuitively, the two following conditions hold:
        \begin{enumerate}
            \item every observable that can compose (under $\sqcap$) with another observable must occur simultaneously with one of its related observables; and
            \item only an observable that does not compose (under $\sqcap$) with any other observable can happen before another observable, i.e., at a strictly lower time.
        \end{enumerate}
        To formalize the conditions above, we use the independence relation $\ind_\sqcap$ where $\ind_\sqcap(X,Y) = \forall x \subseteq X.\forall y \subseteq Y. (x,y) \not\in \sqcap$.

        The \emph{synchronous} composability relation on observations $\kappa^{\ssync}_{\sqcap}(E_1,E_2)$ is the smallest set such that, for $O_1, O_1' \subseteq \Po(E_1)$ and $O_2, O_2' \subseteq \Po(E_2)$:
    \begin{itemize}
        \item if $(O_1, O_2) \in \sqcap$ and $\ind_\sqcap(O_1', E_2)$ and $\ind_\sqcap(E_1, O_2')$ then, for all time stamps $t\in \Rp$, $((O_1\cup O_1', t), (O_2 \cup O_2',t)) \in \kappa_\sqcap^{\ssync}(E_1, E_2)$.
        \item  if $\ind_\sqcap(O_1, E_2)$ then for all $X_2 \subseteq E_2$ and $t_1 \leq t_2$, $((O_1,t_1), (X_2,t_2)) \in \kappa^\ssync_\sqcap(E_1,E_2)$. Reciprocally, if $\ind_\sqcap(E_1, O_2)$ then for all $X_1 \subseteq E_1$ and $t_2 \leq t_1$, $((X_1,t_1), (O_2,t_2)) \in \kappa^\ssync_\sqcap(E_1,E_2)$;
    \end{itemize}

    Let $\bowtie$ be the product defined as $\bowtie\ = \times_{([\kappa_\sqcap^{\mathit{sync}}],[\cup])}$ where $\sqcap = \{ (O,O) \mid O\subseteq \mathbb{E}\}$ with $\mathbb{E}$ the universal set of events. 
    Intuitively, $\bowtie$ synchronizes all observations that contain events shared by the interface of two components.
    We recall that the composition $[\cup]$ of two TESs $\sigma_1$ and $\sigma_2$ is such that it interleaves, in order, all observations with distinct time stamps, and takes the union of two observations with the same time stamp, i.e., the \emph{lifting} of set union $\cup$ to \tes s, written $[\cup]$, is such that, for $\sigma_1,\sigma_2 \in \TES{\E}$ where $\sigma_i(0) = (O_i, t_i)$ with $i \in \{1,2\}$:
    \[
        \sigma_1[\cup]\sigma_2 =
    \begin{cases}
        \langle \sigma_1(0)\rangle \cdot (\sigma_1'[\cup] \sigma_2)   &\mathit{if}\ t_1 < t_2 \\
        \langle \sigma_2(0)\rangle \cdot (\sigma_1[\cup] \sigma_2')     &\mathit{if}\ t_2 < t_1 \\
        \langle (O_1 \cup O_2, t_1)\rangle \cdot (\sigma_1'[\cup] \sigma_2') &  \mathit{otherwise}
    \end{cases}
    \]

    As a result of~\cite{DBLP:journals/corr/abs-2110-02214}, $\bowtie$ is associative and commutative. 
Section~\ref{section:sorting} introduces a motivating example in which robots, roaming on a shared physical medium, must coordinate to sort themselves.
We define algebraically the system consisting of 5 robots and a grid, to which we then add some coordinating protocol components.  

\subsection{Self-sorting robots}
\label{section:sorting}
Consider a robot component $R$, that moves on a grid and reads its position. Its interface consists of the coordinates $(x,y)_R$ that it reads and the direction $d_R$ of its moves, i.e., 
$E_R = \{ (x,y)_R, d_R \mid d \in \{S,W,N,E\},\ (x,y) \in \Rp \times \Rp \}$. 
 The robot freely moves and observes its position at anytime, then its behavior is a subset $L_R \subseteq \TES{E_R}$.
 We use $R$ to refer to an arbitrary robot component with interface $E_R$.

We aim to model the physical constraints explicitly. Thus, the grid on which the robot moves is another component. Its interface is parametrized by a set $I$ of robot identifiers, and by a pair $(n,m)$ for its size on the $x$-axis and $y$-axis respectively, i.e., 
$E_G(I,n,m) = \{ (x,y)_R, d_R \mid d \in \{S,W,N,E\}, R \in I, (x,y) \in [1;n] \times [1;m]\}$.
The behavior of a grid is a set $L_G \subseteq \TES{E_G}$ such that, for any $\sigma \in L_G$ and for any element $\sigma(i) = (O_i, t_i)$, if $(x,y)_R \in O_i$ then $(x,y)_{R'} \not\in O_i$ for every $R' \not = R$: two robots cannot share the same location.
Note that the physics has some internal constraints (no two robots can share the same cell) that no robot \emph{a priori} is aware of. 
    Typically, the move of a robot 
    coincides with a change of state of the grid:
if a robot $R$ is at position $(x,y)_R$ and does a move $N_R$, then the next observable position for robot $R$ will be $(x,y+1)_R$. We give in Example~\ref{ex:TES-tr-sys} an operational definition of the grid component.
We use $G(I, n, m)$ to refer to an arbitrary grid component with set of events $E_G(I,n,m)$.

Figure~\ref{fig:robots} shows five robot instances, each of which has a unique and distinct natural number assigned, positioned at an initial location on a grid.
The goal of the robots in this example is to move around on the grid such that they end up in a final state where they line-up in the sorted order according to their assigned numbers.
    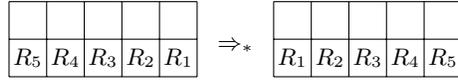
\begin{figure}
        \centering
        \begin{tikzpicture}
            \draw[step=0.5,black,thin] (0,0) grid (2.5,1);
            \node[] at (0.25,0.25) (3) {$R_5$};
            \node[] at (0.75,0.25) (1) {$R_4$};
            \node[] at (1.25,0.25) (2) {$R_3$};
            \node[] at (1.75,0.25) (3) {$R_2$};
            \node[] at (2.25,0.25) (3) {$R_1$};
        \end{tikzpicture}
        \begin{tikzpicture}
            \node[] at (0,0) (1) {};
            \node[] at (0,0.3) (1) {$\Rightarrow_*$};
        \end{tikzpicture}
        \begin{tikzpicture}
            \draw[step=0.5,black,thin] (0,0) grid (2.5,1);
            \node[] at (0.25,0.25) (3) {$R_1$};
            \node[] at (0.75,0.25) (1) {$R_2$};
            \node[] at (1.25,0.25) (2) {$R_3$};
            \node[] at (1.75,0.25) (3) {$R_4$};
            \node[] at (2.25,0.25) (3) {$R_5$};
        \end{tikzpicture}
        \caption{Initial state of the unsorted robots (left), and final state of the sorted robots (right).}
        \label{fig:robots}
    \end{figure}

    \begin{table}
        \centering
        \begin{tabular}{cccccc}
            & $R_1$ & $R_2$ & $R_3$& $R_4$& $R_5$ \\
            \hline 
            $t_1$ & N & - & - & - & - \\   
            $t_2$ & W & - & - & - & - \\   
            $t_3$ & $(3;1)$ & - & - & -& - \\
            $...$ & $...$ & $...$ & $...$& $...$& $...$
        \end{tabular}
        \quad \quad
        \begin{tabular}{cccccc}
            & $R_1$ & $R_2$ & $R_3$ &$R_4$ & $R_5$ \\
            \hline 
            $t_1$ & N & - & - & - & -\\   
            $t_2$ & W & E & N & - & -\\   
            $t_3$ & S & $(4;0)$ & E & - & -\\
            $...$ & $...$ & $...$ & $...$ & $...$ & $...$
        \end{tabular}
        \quad \quad
        \begin{tabular}{cccccc}
            & $R_1$ & $R_2$ & $R_3$ & $R_4$ & $R_5$\\
            \hline 
            $t_1$ & N & - & - & - & N\\   
            $t_2$ & W & E & - & W & E\\   
            $t_3$ & S & - & - & - & S\\
            $...$ & $...$ & $...$ & $...$& $...$ & $...$
        \end{tabular}
        \caption{Each table displays the three first observables at times $t_1$, $t_2$, and $t_3$ for three TESs in the behavior of the product of components $R_1$, $R_2$, $R_3$, $R_4$, and $R_5$ on the grid of Figure~\ref{fig:robots}. We omit the subscript and use the column to identify the events. The symbol - represents the absence of observation in the TES.}
        \label{table:prefix}
    \end{table}

    We write the composite cyber-physical system consisting of $5$ robot components interacting with a shared grid as the following expression:
    \[
        R_1 \bowtie R_2 \bowtie R_3 \bowtie R_4 \bowtie R_5 \bowtie G(\{1,2,3,4,5\}, 5, 2)
    \]

    Three first observations for three behaviors are displayed in Table~\ref{table:prefix}. Each behavior exposes different degrees of concurrency, where in the left behavior, only robot $R_1$ moves, while in the middle behavior, robots $R_1$ and $R_2$ swap their positions, and in the right behavior both $R_1$ and $R_4$ swap their positions with $R_2$ and $R_5$, respectively.

\subsection{Properties of components and coordination}
A component $C = (E,L)$ may satisfy some properties on its behavior. We consider trace properties $P \subseteq \TES{E}$ and say that $C$ satisfies $P$ if and only if $L \subseteq P$, i.e., all the behavior of $C$ is included in the property $P$.

For the set of robots and the grid, we consider the following property: \emph{eventually, the position of each robot $R_i$ is $(i,0)_{R_i}$}, i.e., every robot successfully reaches its place.

This property is a trace property, which we call $P_{\sorted}$ and consists of every behavior $\sigma \in \TES{E}$ such that there exists an $n \in \N$ with $\sigma(n) = (O_n,t_n)$ and $(i,0)_{R_i} \in O_n$ for all robots $R_i$.
As shown in Table~\ref{table:prefix}, the set of behaviors for the product of robots is large, and the property $P_{\sorted}$ does not (necessarily) hold \emph{a priori}: there exists a composite behavior $\tau$ for the component
$
R_1 \bowtie R_2 \bowtie R_3 \bowtie R_4 \bowtie R_5 \bowtie G(\{1,2,3,4,5\}, 5, 2)$ such that $\tau \not \in P_{\sorted}$.

Robots may beforehand decide on some strategies to swap and move on the grid such that their composition satisfies the property $P_\sorted$. 
For instance, consider the following strategy for each robot $R_n$:
    \begin{itemize}
        \item \emph{swapping}: if the last read $(x,y)$ of its location is such that $x < n$, then move North, then West, then South.
        \item \emph{pursuing}: otherwise, move East.
    \end{itemize}

Remember that the grid prevents two robots from moving to the same cell, which is therefore removed from the observable behavior. 
 We emphasize that some sequences of moves for each robot may deadlock, and therefore are not part of the behavior of the system of robots.
 Consider Figure~\ref{fig:moves}, for which each robot follows its internal strategy. 
 Because of non-determinism introduced by the timing of each observations, one may consider the following sequence of observations: first, $R_1$ move North, then West; in the meantime, $R_2$ moves West, followed by $R_3$, $R_4$, and $R_5$. 
 By a similar sequence of moves, the set of robots ends in the configuration on the right of Figure~\ref{fig:moves}.
 In this position and for each robot, the next move dictated by its internal strategy is disallowed, which corresponds to a \emph{deadlock}.
 While behaviors do not contain finite sequences of observations, which makes the scenario of Figure~\ref{fig:moves} not expressible as a TES, such scenario may occur in practice. 
 We give in next Section some analysis to prevent such behavior to happen.
    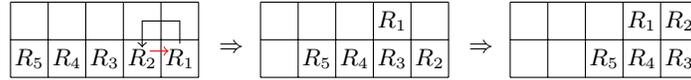
\begin{figure}
        \centering
        \begin{tikzpicture}
            \draw[step=0.5,black,thin] (0,0) grid (2.5,1);
            \node[] at (0.25,0.25) (1) {$R_5$};
            \node[] at (0.75,0.25) (2) {$R_4$};
            \node[] at (1.25,0.25) (3) {$R_3$};
            \node[] at (1.75,0.25) (4) {$R_2$};
            \node[] at (2.25,0.25) (5) {$R_1$};
            \draw[->] (2.25,0.45) to (2.25,0.75) to ++(-0.5,0) to ++(0,-0.35);
            \draw[->, red] (1.85,0.35) to ++(0.25,0);
        \end{tikzpicture}
        \begin{tikzpicture}
            \node[] at (0,0) (1) {};
            \node[] at (0,0.3) (1) {$\Rightarrow$};
        \end{tikzpicture}
        \begin{tikzpicture}
            \draw[step=0.5,black,thin] (0,0) grid (2.5,1);
            \node[] at (1.75,0.75) (1) {$R_1$};
            \node[] at (0.75,0.25) (2) {$R_5$};
            \node[] at (1.25,0.25) (3) {$R_4$};
            \node[] at (1.75,0.25) (4) {$R_3$};
            \node[] at (2.25,0.25) (5) {$R_2$};
        \end{tikzpicture}
        \begin{tikzpicture}
            \node[] at (0,0) (1) {};
            \node[] at (0,0.3) (1) {$\Rightarrow$};
        \end{tikzpicture}
        \begin{tikzpicture}
            \draw[step=0.5,black,thin] (0,0) grid (2.5,1);
            \node[] at (2.25,0.25) (1) {$R_3$};
            \node[] at (2.25,0.75) (2) {$R_2$};
            \node[] at (1.75,0.25) (3) {$R_4$};
            \node[] at (1.25,0.25) (4) {$R_5$};
            \node[] at (1.75,0.75) (5) {$R_1$};
        \end{tikzpicture}
        \caption{Initial state of the unsorted robot (left) leading to a possible deadlock (right) if each robot follows its strategy.}
        \label{fig:moves}
    \end{figure}

Alternatively, the collection of robots may be coordinated by an external protocol that guides their moves. 
Besides considering the robot and the grid components, we add a third kind of component that acts as a coordinator. 
In other words, we make the protocol used by robots to interact explicit and external to them and the grid; i.e., we assume exogenous coordination. 
Exogenous coordination allows robots to decide a priori on some strategies to swap and move on the grid, in which case their external coordinator component merely unconditionally facilitates their interactions.
Alternatively, the external coordinator component may implement a protocol that guides the moves of a set of clueless robots into their destined final locations.
    The most intuitive of such coordinator is the property itself as a component. Indeed, let $C_\sorted = (E,L)$ be such that $E = \bigcup_{i\in I} E_{R_i}$ with $I = \{1,2,3,4,5\}$ and $L = P_\sorted$.
    Then, and as shown in~\cite{DBLP:journals/corr/abs-2110-02214}, the coordinated component 
    $
R_1 \bowtie R_2 \bowtie R_3 \bowtie R_4 \bowtie R_5 \bowtie G(\{1,2,3,4,5\}, 5, 2) \bowtie C_\sorted 
    $
    trivially satisfies the property $P_\sorted$.
    While easily specified, such coordination component is non-deterministic and not easily implementable. 
    We provide an example of a deterministic coordinators.

    Given two robot identifiers $R_1$ and $R_2$, we introduce the swap component $S(R_1, R_2)$ that coordinates the two robots $R_1$ and $R_2$ to swap their positions.
    Its interface $E_S(R_1, R_2)$ contains the following events:
    \begin{itemize}
        \item start(S($R_1$,$R_2$)) and end(S($R_1$, $R_2$)) that respectively notify the beginning and the end of an interaction with $R_1$ and $R_2$. Those events are observed when the swap protocol is starting or ending an interaction with either $R_1$ or with $R_2$.
        \item $(x,y)_{R_1}$ and $(x,y)_{R_2}$ that occur when the protocol reads, respectively, the position of robot $R_1$ and robot $R_2$,
        \item $d_{R_1}$ and $d_{R_2}$ for all $d \in \{N,W,E,S\}$ that occur when the robots $R_1$ and $R_2$ move;
        \item locked(S($R_1$,$R_2$)) and unlocked(S($R_1$,$R_2$)) that occur, respectively, when another protocol begin and end an interaction with either $R_1$ and $R_2$.
    \end{itemize}
    The behavior of a swapping protocol $S(R_1, R_2)$ is such that, it starts its protocol sequence by an observable start(S($R_1, R_2$)), then it moves $R_1$ North, then $R_2$ East, then $R_1$ West and South.
    The protocol starts the sequence only if it reads a position for $R_1$ and $R_2$ such that $R_1$ is on the cell next to $R_2$ on the $x$-axis.
    Once the sequence of moves is complete, the protocol outputs the observable end(S($R_1$, $R_2$)).
    If the protocol is not swapping two robots, or is not locked, then robots can freely read their positions. 
    
    Swapping protocols interact with each others by locking other protocols that share the same robot identifiers. 
    Therefore, if S($R_1$,$R_2$) starts its protocol sequence, then S($R_2$, $R_i$) synchronizes with a locked event locked(S($R_2$,$R_i$)), for $2 < i$. 
    Then, $R_2$ cannot swap with other robots unless S($R_1$,$R_2$) completes its sequence, in which case end(S($R_1$, $R_2$)) synchronizes with unlocked(S($R_2$,$R_i$)) for $2 < i$.
    We extend the underlying composability relation $\sqcap$ on observables such that, for all $(O_1, O_2) \in \sqcap$ and $i<j$:
    \begin{align*}
        \text{start(S($R_i$,$R_j$))} \in O_1 \implies &\exists k. k<i. \text{locked(S($R_k$, $R_i$))}\in O_2 \lor \\ 
                                                      &\exists k. j<k. \text{locked(S($R_j$, $R_k$))} \in O_2
    \end{align*}
    and
    \begin{align*}
        \text{end(S($R_i$,$R_j$))} \in O_1 \implies &\exists k<i. \text{unlocked(S($R_k$, $R_i$))}\in O_2 \lor \\ 
                                                      &\exists j<k. \text{unlocked(S($R_j$, $R_k$))} \in O_2
    \end{align*}

    For each pair of robots $R_i$, $R_j$ such that $i<j$, we introduce a swapping protocol S($R_i$, $R_j$). As a result, the coordinated system is given by the following composition:
$$
R_1 \bowtie R_2 \bowtie R_3 \bowtie R_4 \bowtie R_5 \bowtie G(\{1,2,3,4,5\}, 5, 2) \bowtie_{i<j} S(R_i, R_j)
$$

Note that the definition of $\sqcap$ and of $\kappa_\sqcap^\ssync$ impose that, if one protocol starts its sequence, then all protocols that share some robot identifiers synchronize with a lock event. Similar behavior occurs at the end of the sequence. 

The study of the coordinated system, and the analysis of its possible deadlock is the object of Section~\ref{section:operational}. We give an operational specification of components using TES transition systems. We define the possibility for such operational components to deadlock, and give some conditions to prevent deadlock from occurring given a composition of deadlock free components.

\section{An operational specification of components}
\label{section:operational}

In Section~\ref{section:components}, we give a declarative specification of components, and considers infinite behaviors only. 
We give, in Section~\ref{section:TES-tr-sys}, an operational specification of components using TES transition systems.
We relate the parametrized product of TES transition systems with the parametrized product on their corresponding components, and show its correctness. 
The composition of two TES transition systems may lead to transitions that are not composable, and ultimately to a deadlock, i.e., a state with no outgoing transitions. 

\subsection{TES transition systems.}
\label{section:TES-tr-sys}
The behavior of a component as in Definition~\ref{def:component} is a set of TESs. 
We give an operational definition of such set using a labelled transition system.
\begin{definition}[TES transition system]
A TES transition system is a triple $(Q,E,\rightarrow)$ where $Q$ is a set of states, $E$ is a set of events, and $\rightarrow \subseteq Q \times (\Po(E) \times \Rp) \times Q$ is a labeled transition relation, where labels are observations.
\hfill$\triangle$
\end{definition}

We present two different ways to give a semantics to a TES transition system: inductive and co-inductive. 
Both definitions give the same behavior, as shown in Theorem~\ref{thm:semantics}, and we use interchangeably each definition to simplify the proofs of, e.g., Theorem~\ref{thm:correctness}.

\paragraph{Semantics 1 (runs).}
Let $T = (Q,E,\rightarrow)$ be a TES transition system.
Given $s \in (\Po(E)\times \Rp)^n$, we write $q \xrightarrow{s} p$ for the sequence of transitions $q \xrightarrow{s(0)} q_1 \xrightarrow{s(1)} q_2\ ... \xrightarrow{s(n)} p$. We use $\to^*$ and $\to^\omega$ to denote, respectively, the set of finite and infinite sequences of consecutive transitions in $\to$.
Then, finite sequences of observables form the set $\Fin{T,q} = \{ \sigma \in \TES{E} \mid q \xrightarrow{s} q', \exists n. s = \sigma[n] \land \sigma^{(n)} \in \TES{\emptyset} \}$ and infinite ones, the set $\Inf{T,q} = \{ \sigma \in \TES{E} \mid \forall n. \sigma[n] \in \Fin{T,q} \}$ where, as introduced above, $\sigma[n]$ is the prefix of size $n$ of $\sigma$.
The semantics of such a TES transition system $T = (Q,E,\rightarrow)$, starting in a state $q \in Q$, is the component $C_T(q) = (E, \Inf{T,q})$. 


\paragraph{Semantics 2 (greatest post fixed point)}
Alternatively, the semantics of a TES transition system is the greatest post fixed point of a function over sets of TESs paired with a state. 
For a TES transition system $T = (Q, E, \to)$, let $\mathcal{R} \subseteq \TES{E} \times Q$. We introduce $\phi_{T} : \Po(\TES{E}\times Q) \rightarrow \Po(\TES{E}\times Q)$
as the function:
$$
\begin{array}{rl}
    \phi_T(\Rel) = \{ (\tau,q) \mid   & \exists p \in Q,\ q\xrightarrow{\tau(0)}p \land (\tau',p) \in \Rel\}
\end{array}
$$

We can show that $\phi_T$ is monotonous, and therefore $\phi_T$ has a greatest post fixed point $\Omega_T = \bigcup \{ \Rel \mid \Rel \subseteq \phi_T(\Rel)\}$.
We write $\Omega_T(q) = \{ \tau \mid (\tau,q) \in \Omega_T\}$ for any $q \in Q$.
Note that the two semantics coincide.
\begin{theorem}[Equivalence]
    \label{thm:semantics}
    For all $q \in Q$, $\Inf{T,q} = \{ \tau \mid (\tau, q) \in \Omega_T\}$.
\end{theorem}
\begin{proof}
    Let $T = (Q,E\to)$ and $\Omega_T(q) = \{ \tau \mid (\tau,q) \in \Omega_T\}$.
    \[
\begin{array}{rl}
    \tau \in \Omega_{T}(c_0) 
\iff &  (\tau, c_0) \in \Omega_T \\
\iff &  c_0 \xrightarrow{\tau(0)}c_1 \land (\tau',c_1) \in \Omega_T\\
\iff &  \exists \chi \in \rightarrow^\omega.\  \chi(0) = c_0\xrightarrow{\tau(0)} c_1 \land (\tau', c_1) \in \Omega_T \\
\iff & \exists \chi \in \rightarrow^\omega, c\in Q^\omega.\ c(0) \in Q_0 \land \forall n \in \N. \\ 
     & \chi(n) = c(n)\xrightarrow{\tau(n)} c(n+1) \land (\tau^{(n+1)},c(n+1)) \in \Omega_T \\
\iff & \exists \chi \in \to^\omega.\ \chi(0) = c_0 \xrightarrow{\tau(0)}c_1\land \tau \in \Inf{T,q} \\
                        \iff & \tau \in \Inf{T,q} 
\end{array} 
\]
In the fourth equivalence, we state that the infinite sequence $\chi \in \rightarrow^\omega$ has, as sequence of labels, the sequence of observations in $\tau$. 
We prove the step by induction.
Let $n \in \N$, and let $\chi \in \rightarrow^\omega$ be such that, for all $k\leq n$, $\chi(k) = c_k \xrightarrow{\tau(k)} c_{k+1}$ with $c_{k} \in Q$ and $(\tau^{(k)}, c_k) \in \Omega_T$. 
Then, given that $(\tau^{(n)}, c_n) \in \Omega_T$, there exists a transition $c_n \xrightarrow{\tau(n+1)} c_{n+1}$ and there exists $\rho \in \to^\omega$ such that, for all $k \leq n$, $\rho = \chi$ and $r(n+1) = c_n \xrightarrow{\tau(n+1)} c_{n+1}$, which proves the implication.
The other direction of the equivalence is simpler. If there exists $\chi \in \to^\omega$ such that for all $n \in \N$, $\chi(n) = c_n \xrightarrow{\tau(n)} c_{n+1}$, then we have a witness, for every $n \in \N$, that $(\tau^{(n)}, c_n)$ is an element of $\Omega_T$.
\qed
\end{proof}

\begin{remark}[Deadlock]
    Observe that $\FG{\Inf{T,q}} \subseteq \Fin{T,q}$ which, in the case of strict inclusion, captures the fact that some states may have no outgoing transitions and therefore deadlock. 
\end{remark}

\begin{remark}[Expressivness]
    There may be two different TES transition systems $T_1$ and $T_2$ such that $\Inf{T_1} = \Inf{T_2}$, i.e., a set of TESs is not uniquely characterized by a TES transition system.
\end{remark}
\newcommand{\Lock}{\mathit{lock}}
\newcommand{\Unlock}{\mathit{unlock}}
\newcommand{\Start}{\mathit{start}}
\newcommand{\End}{\mathit{end}}
\begin{example}
    \label{ex:TES-tr-sys}
    The behavior of a robot introduced earlier is a TES transition system $T_R = (\Rp,E_R, \rightarrow)$ where 
    $t \xrightarrow{(\{e\}, t)} t + \delta$ for abitrary $t$ and $\delta$ in $\Rp$ and $e \in E_R$.

    Similarly, the behavior of a grid is a TES transition system $T_G(I,n,m) = (Q_G,E_G(I,n,m), \rightarrow)$ where:
    \begin{itemize}
        \item $Q_G \subseteq \Rp\times (I\to ([0;n]\times[0;m]))$,
        \item $(t, f) \xrightarrow{(O, t)} (t + \delta, f')$ for abitrary $t$ and $\delta$ in $\Rp$, such that
            \begin{itemize}
                \item $d_R \in O$ implies $f'(R)$ is updated according to the direction $d$ if the resulting position is within the bounds of the grid;
                \item $(x,y)_R \in O$ implies $f(R) = (x,y)_R$ and $f'(R) = f(R)$;
                \item $f'(R) = f(R)$, otherwise.
            \end{itemize}
    \end{itemize}

    The behavior of a swap protocol S($R_i$,$R_j$) with $i<j$ is a TES transition system $T_S(R_1, R_2) = (Q, E, \to)$ where, for $t_1, t_2, t_3, t_4 \in \Rp$ with $t_1 < t_2 < t_3 < t_4$:
    \begin{itemize}
        \item $Q \subseteq \Rp \times \{s_1, s_2, s_3, s_4,s_6\}$;
        \item $E = E_{R_i} \cup E_{R_j} \cup \{\Lock(R_i,R_j), \Unlock(R_i,R_j), \Start(R_i,R_j), \End(R_i,R_j)\}$
        \item $(t_1, s_1) \xrightarrow{(\{\Lock(R_i,R_j)\}, t_1)} (t_2, s_2)$; 
        \item $(t_1, s_2) \xrightarrow{(\{\Unlock(R_i,R_j)\}, t_1)} (t_2, s_1)$; 
        \item $(t_1, s_1) \xrightarrow{(\{\Start(R_i,R_j), (x,y)_{R_i}, (x+1,y)_{R_j}\}, t_1)} (t_2, s_3)$;
        \item $(t_1, s_3)\xrightarrow{(\{N_{R_j}\}, t_1)}(t_2, s_4)\xrightarrow{(\{\{W_{R_j}, E_{R_i}\}\}, t_2)} (t_3, s_5)\xrightarrow{(\{\{S_{R_j}\}\}, t_3)} (t_4, s_6)$; 
        \item $(t_1, s_6) \xrightarrow{(\{\End(R_i,R_j)\}, t_1)} (t_2, s_1)$; 
    \end{itemize}
    \hfill$\blacksquare$
\end{example}

The product of two components is parametrized by a composability relation and a composition function and syntactically constructs the product of two TES transition systems.

\begin{definition}[Product]
The product of two TES transition systems $T_1 = (Q_1, E_1, \to_1)$ and $T_2 = (Q_2, E_2, \to_2)$ under the constraint $\kappa$ is the TES transition system $T_1 \times_\kappa T_2 = (Q_1 \times Q_2, E_1 \cup E_2, \to)$ such that:
\[
\cfrac{
    q_1 \xrightarrow{(O_1, t_1)}_1 q_1' \quad q_2 \xrightarrow{(O_2, t_2)}_2 q_2' \quad ((O_1, t_1), (O_2, t_2)) \in \kappa(E_1, E_2) \quad t_1 < t_2
}{
(q_1,q_2) \xrightarrow{(O_1, t_1)} (q_1',q_2)
}
\]
\[
\cfrac{
    q_1 \xrightarrow{(O_1, t_1)}_1 q_1' \quad q_2 \xrightarrow{(O_2, t_2)}_2 q_2' \quad ((O_1, t_1), (O_2, t_2)) \in \kappa(E_1, E_2) \quad t_2 < t_1
}{
(q_1,q_2) \xrightarrow{(O_2, t_2)} (q_1,q_2')
}
\]
\[
\cfrac{
    q_1 \xrightarrow{(O_1, t_1)}_1 q_1' \quad q_2 \xrightarrow{(O_2, t_2)}_2 q_2' \quad ((O_1, t_1), (O_2, t_2)) \in \kappa(E_1, E_2) \quad t_1 = t_2
}{
(q_1,q_2) \xrightarrow{(O_1\cup O_2, t_1)} (q_1',q_2')
}
\]
\hfill$\triangle$
\end{definition}

Observe that the product is defined on pairs of transitions, which implies that if $T_1$ or $T_2$ has a state without outgoing transition, then the product has no outgoing transitions from that state. The reciprocal is, however, not true in general.

Theorem~\ref{thm:correctness} states that the product of TES transition systems denotes
 (given a state) the set of TESs that corresponds to the product of the
 corresponding components (in their respective states).
 Then, the product that we define on TES transition systems does not add nor
 remove behaviors with respect to the product on their respective components.
 \begin{theorem}[Correctness]
    \label{thm:correctness}
    For all TES transition systems $T_1$ and $T_2$, and for all composability relation $\kappa$:
    \[C_{T_1 \times_\kappa T_2}(q_1,q_2) = C_{T_1}(q_1) \times_{([\kappa],[\cup])} C_{T_2}(q_2)\]
\end{theorem}
\begin{proof}
    Let $T_1 = (Q_1, E_1, \to_1)$ and $T_2 = (Q_2, E_2, \to_2)$.
The proof goes in two directions:
\begin{enumerate}
    \item We first show that, for any $\tau$ in the behavior of $\Inf{T_1 \times_{\kappa} T_2, (q_1,q_2)}$, there exist $\tau_1 \in \Inf{T_1,q_1}$ and $\tau_2 \in \Inf{T_2,q_2}$ such that $\tau = \tau_1[\cup] \tau_2$ and $(\tau_1, \tau_2) \in [\kappa](E_1, E_2)$.
    \item We then show that, for any $\tau_1 \in \Inf{T_1,q_1}$ and $\tau_2 \in \Inf{T_2,q_2}$ such that $(\tau_1, \tau_2) \in [\kappa](E_1, E_2)$, we have $(\tau_1 [\cup] \tau_2) \in \Inf{T_1 \times_\kappa T_2,(q_1, q_2)}$.
\end{enumerate}

We recall the definition of $\Phi_\kappa : \Po(\TES{E} \times \TES{E}) \rightarrow \Po(\TES{E} \times \TES{E})$, the function defining the lifting of $\kappa$ from observations to TESs (by co-induction), to be such that, for all $\Rel \subseteq  \TES{E} \times \TES{E}$:
$$
\begin{array}{rl}
    \Phi_\kappa(\Rel)= \{ (\tau_1, \tau_2) \mid 
    & \tau_1(0) = (O_1,t_1) \land \tau_2(0) = (O_2,t_2)\ \land  \\
                                                &(\tau_1(0), \tau_2(0)) \in \kappa(E_1, E_2) \land [ t_1 < t_2 \land (\tau_1' , \tau_2) \in \Rel \lor \\
                                                & \hskip12em\  t_2 < t_1 \land (\tau_1, \tau_2') \in \Rel \lor   \\
                                                & \hskip12em\  t_1 = t_2 \land (\tau_1',\tau_2') \in \Rel ]
            \}
\end{array}
$$

First, observe that, by construction, for all $\tau \in \Omega_{T_1 \times_\kappa T_2}(q_1,q_2)$, there exist $\tau_1 \in \Omega_{T_1}(q_1)$ and $\tau_2 \in \Omega_{T_2}(q_2)$ such that $\tau = \tau_1 [\cup] \tau_2$. Indeed, each transition in $T_1\times_\kappa T_2$ is constructed out of a composable pair of observations from a transition in $T_1$ and in $T_2$. 
The resulting label is identical to the label that $[\cup]$ defines co-inductively.
Moreover, each element $\tau \in \Omega_{T_1 \times_\kappa T_2}(q_1,q_2)$ contains infinitely many labels from $T_1$ and from $T_2$, due to the assumption that a TES has an increasing an non-Zeno sequence of time stamps.

We therefore use $\widetilde{\Omega}_{T_1 \times_\kappa T_2}(q_1,q_2)$ to denote the set of such pairs.
We then prove that, for all $(q_1,q_2) \in Q_1\times Q_2$, $(\tau_1, \tau_2) \in \widetilde{\Omega}_{T_1 \times_\kappa T_2}(q_1, q_2)$ implies that $(\tau_1, \tau_2) \in [\kappa](E_1, E_2)$ and $(\tau_1, \tau_2) \in \Omega_{T_1}(q_1) \times \Omega_{T_2}(q_2)$.

We prove that $\widetilde{\Omega}_{T_1 \times_\kappa T_2}(q_1,q_2) = \{(\tau_1, \tau_2) \in [\kappa](E_1, E_2) \mid (\tau_1, \tau_2) \in \Omega_{T_1}(q_1) \times \Omega_{T_2}(q_2) \}$ by showing forward and backward inclusion, i.e., point $1$ and point $2$ respectively.

\underline{Forward inclusion.} 
We know that since $(\tau_1, \tau_2) \in \widetilde{\Omega}_{T_1\times_\kappa T_2}(q_1, q_2)$, then there exists a post fixed point $R$ of $\Phi_{T_1 \times_\kappa T_2}$ such that $((\tau_1, \tau_2), (q_1, q_2)) \in R$. Let $\widetilde{R}$ be the set $\{ (\tau_1, \tau_2) \mid  \exists q. ((\tau_1, \tau_2),q) \in R\}$.
We show that $\widetilde{R}$ is a post fixed point of $[\kappa](E_1, E_2)$.
By definition of the transition relation $\to$ of $T_1 \times_\kappa T_2$, we have that if $(\tau_1, \tau_2) \in \widetilde{R}$, then $(\tau_1(0), \tau_2(0)) \in \kappa(E_1, E_2)$ and, if $t_1 < t_2$ then $(\tau_1' , \tau_2) \in \widetilde{R}$, if $t_2 < t_1$ then $(\tau_1, \tau_2') \in \widetilde{R}$, and if $t_1 = t_2$ then $(\tau_1',\tau_2') \in \widetilde{R}$. We thus conclude that $\widetilde{R} \subseteq \Phi_\kappa(\widetilde{R})$, and therefore $\widetilde{\Omega}_{T_1\times_\kappa T_2}(q_1,q_2) \subseteq [\kappa](E_1, E_2)$.

We then showed that, for any $\tau$ in the behavior of $\Inf{T_1 \times_{\kappa} T_2, (q_1,q_2)}$, there exist $\tau_1 \in \Inf{T_1,q_1}$ and $\tau_2 \in \Inf{T_2,q_2}$ such that $\tau = \tau_1[\cup] \tau_2$ and $(\tau_1, \tau_2) \in [\kappa](E_1, E_2)$.

\medskip

\underline{Backward inclusion.}
We show that, for all $(q_1,q_2) \in Q_1\times Q_2$, if $(\tau_1, \tau_2) \in [\kappa](E_1, E_2)$ and $(\tau_1, \tau_2) \in \Omega_{T_1}(q_1) \times \Omega_{T_2}(q_2)$ then $(\tau_1, \tau_2) \in \widetilde{\Omega}_{T_1 \times_\kappa T_2}(q_1, q_2)$.

To do so, we construct a post fixed points of $\Phi_\kappa$, which we denote $S$, such that $((\tau_1, \tau_2), (q_1, q_2)) \in S$; and, for all  $((\tau, \sigma), (q, s)) \in S$, we have $\tau(0) = (O_1, t_1)$ and $\sigma(0) = (O_2, t_2)$: 
\begin{itemize}
    \item if $t_1 < t_2$ and $(q_1, q_2) \xrightarrow{\tau(0)} (q_1', q_2)$ then $((\tau',\sigma), (q_1', q_2)) \in S$;
    \item if $t_2 < t_1$ and $(q_1, q_2) \xrightarrow{\sigma(0)} (q_1, q_2')$ then $((\tau,\sigma'), (q_1, q_2')) \in S$; and
    \item if $t_2 = t_1$ and $(q_1, q_2) \xrightarrow{\sigma(0)\cup \tau(0)} (q_1', q_2')$ then $((\tau',\sigma'), (q_1', q_2')) \in S$.
\end{itemize}
Since $(\tau_1, \tau_2) \in [\kappa](E_1, E_2)$, we know that $S \subseteq [\kappa](E_1, E_2)$.
We now show that $S \subseteq \Phi_{T_1 \times_\kappa T_2}(S)$.
The argument for $S$ being a post fixed point of $\Phi_{T_1 \times_\kappa T_2}$ follows from the definition of $S$.

    \qed
\end{proof}

We give in Example~\ref{ex:prod-TES-tr-sys} the TES transition systems resulting from the product of the TES transition systems of two robots and a grid.
Example~\ref{ex:prod-TES-tr-sys} defines operationally the components in Section 2.2, i.e., their behavior is generated by a TES transition system.
\begin{example}
    \label{ex:prod-TES-tr-sys}
    Let $T_{R_1}$, $T_{R_2}$ be two TES transition systems for robots $R_1$ and $R_2$, and let $T_G(\{1\},n,m)$ be a grid with robot $R_1$ alone and $T_G(\{1, 2\}, n,m)$ be a grid with robots $R_1$ and $R_2$.
    Let $E = E_1 \cup E_2 \cup E_G$ and $\sqcap \subseteq \Po(E) \times \Po(E)$ be such that $(O_1,O_2) \in \sqcap$ if and only if $O_1=O_2 \subseteq E_R \cap E_G$. 
    We use $\kappa_\sqcap^{\mathit{sync}}$ as defined in~\cite{DBLP:journals/corr/abs-2110-02214}.


    The product of $T_{R_1}$, $T_{R_2}$, and $T_G(\{1, 2\}, n, m)$ under $\kappa_\sqcap^\mathit{sync}$ is the TES transition system $T_{R_1}\times_{\kappa_\sqcap^\mathit{sync}} T_{R_2} \times_{\kappa_\sqcap^\mathit{sync}} T_G(\{1, 2\}, n, m)$ such that it synchronizes observations of the two robots with the grid, but does not synchronize events of the two robots directly, since the two set of events are disjoint.

    \hfill $\blacksquare$
\end{example}

As a consequence of Theorem 1, letting $\kappa_\sqcap^{\mathit{sync}}$ be the composability relation used in the product $\bowtie$ and writing $T = T_{R_1} \times_{\kappa_\sqcap^{\mathit{sync}}} T_{R_2} \times_{\kappa_\sqcap^{\mathit{sync}}} T_G$, $C_{T}(q1, q2, q3)$ is equal to the component $C_{T_{R_1}}(q_1) \bowtie C_{T_{R_2}}(q_2) \bowtie C_{T_G}(q_3)$

Let $T$ be a TES transition system, and let $C_T(q) = (E, \Inf{T,q})$ be a component whose behavior is defined by $T$.
Then, $C$ is \emph{deadlock free} if and only if $\FG{\Inf{T,q}} = \Fin{T,q}$.

A class of deadlock free components is that of components that accept arbitrary insertion of $\emptyset$ observables in between two observations. We say that such component is \emph{prefix-closed}, as every sequence of finite observations can be continued by an infinite sequence of empty observables, i.e., $C$ is such that $C = C^*$ (as defined after Definition~\ref{def:prod}).
We say that a TES transition system $T$ is prefix-closed if and only if and only if $C_T(q) = C^*_T(q)$.
For instance, if $T$ is such that, for any state $q$ and for any $t\in\Rp$ there is a transition $q\xrightarrow{(\emptyset,t)}q$, then $T$ is prefix-closed.
\begin{lemma}
    \label{lemma:prefix-closed-prod}
    If $T_1$ and $T_2$ are prefix-closed, then $T_1\times_{\kappa_\sqcap^\ssync} T_2$ is prefix-closed where $\sqcap\subseteq \Po(E_1)\times \Po(E_2)$ with $(O_1,O_2)\in\sqcap$ implies $O_1\not = \emptyset$ and $O_2 \not = \emptyset$.
\end{lemma}
\begin{proof}
    The proof follows from the fact that $\emptyset$ is independent with any non-empty observable $O\subseteq E_1 \cup E_2$.
    Then, any pair of silent observation is composable, and therefore the following TES transition system is prefix-closed.
    \qed
\end{proof}

We search for the condition under which deadlock freedom is preserved under a product.
Section 3.3 gives a condition for the product of two deadlock free components to be deadlock free.

\subsection{Compatibility of components}

Informally, the condition of $\kappa$-compatibility from two TES transition systems $T_1$ and $T_2$ translates the existence of a relation $\mathcal{R}$ on pairs of states of $T_1$ and $T_2$ such that $(q_1, q_2) \in \R$ and for every state $(q, s) \in \mathcal{R}$, every outgoing transition from $T_1$ (reciprocally $T_2$) has a transition in $T_2$ that composes under $\kappa$.
In all cases, the pair of outgoing states is in the relation $\mathcal{R}$.

Formally, a TES transition system $T_1 = (Q_1, E_1, \to_1)$ from state $q_1$ is $\kappa$-compatible with a TES transition system $T_2 = (Q_2, E_2, \to_2)$ from state $q_2$, and we say $(T_1, q_1)$ is $\kappa$-compatible with $(T_2, q_2)$, if there exists a relation $\R \subseteq Q_1 \times Q_2$ such that, $(q_1,q_2) \in \R$ and for any $(p_1, p_2) \in \R$, 
\begin{itemize}
    \item there exist $p_1 \xrightarrow{(O_1, t_1)}_1 r_1$  and $p_2 \xrightarrow{(O_2, t_2)}_2 r_2$ such that $((O_1, t_1),(O_2, t_2)) \in \kappa(E_1, E_2)$; and
    \item for all $p_1 \xrightarrow{(O_1, t_1)}_1 r_1$ and $p_2 \xrightarrow{(O_2, t_2)}_2 r_2$ if $((O_1, t_1),(O_2, t_2)) \in \kappa(E_1, E_2)$ then $(u_1, u_2) \in \R$, where $u_i = r_i$ if $t_i = \min\{t_1, t_2\}$, and $u_i = p_i$ otherwise for $i \in \{1,2\}$.
\end{itemize}

In other words, if $(T_1,q_1)$ is $\kappa$-compatible with $(T_2,q_2)$, 
then all transitions in $T_1$ starting in $q_1$ eventually match some composable transitions in $T_2$ starting in $q_2$.
If $(T_2,q_2)$ is $\kappa$-compatible to $(T_1,q_1)$ as well, then we say that $(T_1,q_1)$ and $(T_2,q_2)$ are $\kappa$-compatible.


\begin{theorem}[Deadlock free]
    Let $(T_1,q_1)$ and $(T_2,q_2)$ be $\kappa$-compatible. Let $C_{T_1}(q_1)$ and $C_{T_2}(q_2)$ be deadlock free.
    Then, $C_{T_1}(q_1) \times_{([\kappa],[\cup])} C_{T_2}(q_2)$ is deadlock free.
\end{theorem}
\begin{proof}
    We reason by contradiction.
    If the product  $C_{T_1}(q_1) \times_{([\kappa],[\cup])} C_{T_2}(q_2)$ 
    is not deadlock free, then 
    $\Fin{T_1\times_\kappa T_2, (q_1,q_2)} \not = \FG{\Inf{T_1\times_\kappa T_2, (q_1,q_2)}}$.
    Thus, there exists a state $(s,q)$, reachable from $(q_1, q_2)$, such that 
    $\Fin{T_1\times_\kappa T_2, (s,q)} = \emptyset$.
    Given the fact that both TES transition systems are deadlock free, and given that $s$ (respectively $q$) is reachable from $q_1$ (respectively $q_2$) for $T_1$ (respectively $T_2$), then 
    $\Fin{T_2, q_1} \not= \emptyset$ and
    $\Fin{T_1, q_2} \not= \emptyset$.

    Since $(T_1,q_1)$ and $(T_2,q_2)$ are $\kappa$-compatible, then there exists $\Rel$ such that for each pair $(s,q) \in \Rel$, there exists an outgoing transition in $T_1$ and $T_2$ from $s$ and $q$ respectively that is composable under $\kappa$.
Such property would imply that there is a transition in $T_1 \times_\kappa T_2$ from state $(s,q)$ and therefore 
    $\Fin{T_1\times_\kappa T_2, (s,q)} \not= \emptyset$.
    In other words, the property of compatibility contradicts the presence of deadlock in the product $C_{T_1}(q_1) \times_{([\kappa],[\cup])} C_{T_2}(q_2)$.
    \qed
\end{proof}

\begin{lemma}
    If $T_1 = (Q_1, E_1, \rightarrow_1)$ and $T_2 = (Q_1, E_1, \rightarrow_1)$ are deadlock free and are such that $\sqcap \cap (\Po(E_1)\times \Po(E_2)) = \emptyset$, then $(T_1, q_1)$ is $\kappa_\sqcap^\mathit{sync}$-compatible with $(T_2, q_2)$ for arbitrary $(q_1, q_2)\in Q_1\times Q_2$.
\end{lemma}
\begin{proof}
    If the relation $\sqcap$ does not relate any observables, then any pair of observations is composable under $\kappa^\ssync_\sqcap$.
    As a consequence, any sequence of transitions is allowed.
    \qed
\end{proof}
\begin{lemma}
    If $T_1$ and $T_2$ are prefix-closed, then $(T_1,q_1)$ and $(T_2,q_2)$ are $\kappa^\ssync_\sqcap$-compatible for arbitrary $q_1\in Q_1$, $q_2\in Q_2$, and $\sqcap \subseteq \Po(E_1)\times \Po(E_2)$ with $(O_1,O_2)\in\sqcap$ implies $O_1 \not = \emptyset$ and $O_2 \not = \emptyset$.
\end{lemma}
\begin{proof}
    Follows from Lemma~\ref{lemma:prefix-closed-prod}.
\end{proof}


The consequence of two TES transition systems $T_1$ and $T_2$ to be $\kappa$-compatible on $(q_1, q_2)$ and deadlock free, is that they can be run \emph{step-by-step} from $(q_1, q_2)$ (i.e., the product can be done at runtime), and the resulting behavior is an element of $\Inf{T_1\times_\kappa T_2, (q_1, q_2)}$.
In general however,  $\kappa$-compatibility is not preserved over product, demonstrated by Example~\ref{ex:counter-ex}.
For the case of coordinated cyber-physical systems, components are usually not prefix-closed as there might be some timing constraints or some mandatory actions to perform in a bounded time frame.
\begin{example}
    \label{ex:counter-ex}
    Suppose three TES transition systems $T_i = (\{q_i\},\{a,b,c,d\},\to_i)$, with $i \in \{1,2,3\}$, defined as follow for all $n \in \N$:
    \begin{itemize}
        \item $q_1 \xrightarrow{(\{a,b\},n)}_1 q_1$ and  $q_1 \xrightarrow{(\{a,c\},n)}_1 q_1$;
        \item $q_2 \xrightarrow{(\{a,c\},n)}_2 q_2$ and  $q_2 \xrightarrow{(\{a,d\},n)}_2 q_2$;
        \item $q_3 \xrightarrow{(\{a,d\},n)}_3 q_3$ and  $q_3 \xrightarrow{(\{a,b\},n)}_3 q_3$.
    \end{itemize}
    Let $\sqcap = \{(O,O) \mid O \subseteq \{a,b,c,d\}\}$. It is easy to show that $T_1(q_1)$, $T_2(q_2)$, and $T_3(q_3)$ are pairwise $\kappa^\ssync_\sqcap$-compatible.
    However, $T_1(q_1)$ is not $\kappa^\ssync_\sqcap$-compatible with $T_2(q_2)\times_\sqcap^\ssync T_3(q_3)$.

    \hfill$\blacksquare$
\end{example}

Note that, if $(T_1,q_1)$, $(T_2, q_2)$, and $(T_3, q_3)$ are pairwise $\kappa$-compatible TES transition systems, and if one of the following holds:
\begin{itemize}
    \item $T_1$, $T_2$, $T_3$ are all three prefix-closed TES transition systems;
    \item $\sqcap \cap (\Po(E_1)\times\Po(E_2) \cup  \Po(E_2)\times\Po(E_3) \cup \Po(E_1)\times\Po(E_3)) = \emptyset$; or
    \item $T_1$, $T_2$, $T_3$ are such that, for all of their state, there is one and only one outgoing transition.
\end{itemize}
    then $(T_1,q_1)$ and $(T_2 \times_\kappa T_3, (q_2,q_3))$ are pairwise $\kappa^\ssync_\sqcap$-compatible.

\section{Application: self-sorting robots} 
We implemented in Maude a framework to simulate concurrent executions of TES transition systems, where time stamps are restricted to natural numbers.
Using the description given in Example~\ref{ex:TES-tr-sys} for the grid and for robots, we add to their composition several protocols that aim at preventing deadlock.
The source for the implementation is accessible at~\cite{cp-agents} to reproduce the results of this section.

\paragraph{Components in Maude}
The implementation of TES transition systems in Maude focuses on a subset that has some properties.
First, TES transition systems in Maude have time stamps that range over the set of positive natural numbers $\N$. We do not implement components with real time. 

Second, TES transition systems are delay insensitive. This property encodes that arbitrary time may pass for each transition in the TES transition system. Then, a TES transition system that is delay insensitive is such that if $q \xrightarrow{(O,n)} p$, then $q \xrightarrow{(O,n+k)} p$ for arbitrary $k\in\N$.
We therefore write $q \xrightarrow{O} p$ to denote the set of transitions $q \xrightarrow{(O,n)} p$ for all $n \in \N$. 


In Maude, the state of a TES transition system component is represented by a term and the state of a composed system is a multiset of component states.   Transitions of the step-wise product are defined in terms of such system states.
For instance, the swap protocol between robot $R(3)$ and $R(1)$ is the following term in Maude:
{\tiny
\begin{lstlisting}
 [swap(R(3),R(1)): Protocol | k("state") |-> ds(q(0)); false; null]
 \end{lstlisting}
 }
 where swap(R(3),R(1)) is the name of the component; Protocol is its class; k(``state'') maps to the initial state of the protocol q(0); ``false'' denotes the status of the protocol; and ``null'' is the set of transitions that the protocol may take.

\paragraph{Runtime composition.}
The product of TES transition systems is constructed at runtime, step by step.
We fix a composability relation $\kappa_\sqcap^\mathit{sync}$ to be such that, for all $O \subseteq \mathbb{E}$, $(O,O) \in \sqcap$.
We use $\kappa_\sqcap^\mathit{sync}$ for the product of TES transition systems.

Given a list of initialized TES transition system, the runtime computes the set of all possible composite transitions, from which transitions that violate the composability relation $\kappa_\sqcap^\mathit{sync}$ are filtered out, and one transition that is composable is non-deterministically chosen.
\begin{algorithm}
    \caption{Runtime composition}
    \begin{algorithmic}[1]
            \Require 
            \Statex   - $n$ initialized TES transition systems $S = \{T_1(q_1), ..., T_n(q_n)\}$ 
            \Statex   - composability relation $\sqcap$ over observables
            \Procedure{RuntimeComposition}{}
            \For{$T_i(q_i) \in S$} 
            \State add $\{q_i\xrightarrow{O_i}_i p_i \mid p_i \in Q_i\}$ to $\textit{Tr}$
            \EndFor
            \While{$trs_i, trs_j \in \textit{Tr}$}
            \For{$q_i\xrightarrow{O_i}p_i \in trs_i$ and $q_j \xrightarrow{O_j}p_j \in trs_j$}
            \If {$((O_i, 1), (O_j, 2)) \in \kappa^\mathit{sync}_\sqcap(E_i, E_j)$}
            \State  add $(q_i,q_j)\xrightarrow{O_i}(p_i,q_j)$ to $\textit{trs}_{ij}$ 
            \EndIf
            \If {$((O_i, 2), (O_j, 1)) \in \kappa^\mathit{sync}_\sqcap(E_i, E_j)$}
            \State  add $(q_i,q_j)\xrightarrow{O_j}(q_i,p_j)$ to $\textit{trs}_{ij}$ 
            \EndIf
            \If {$((O_i, 1), (O_j, 1)) \in \kappa^\mathit{sync}_\sqcap(E_i, E_j)$}
            \State  add $(q_i,q_j)\xrightarrow{O_i\cup O_j}(p_i,p_j)$ to $\textit{trs}_{ij}$ 
            \EndIf
            \EndFor
            \State $\textit{Tr} := (\textit{Tr} \setminus\{\textit{trs}_i, \textit{trs}_i\}) \cup \{\textit{trs}_{ij}\}$
            \EndWhile
            \State let $trs \in \textit{Tr}$
            \State let $(q_1,..., q_n) \xrightarrow{O} (r_1, ...., r_n) \in trs$
            \For{$i \leq n$}
            \State   $T_i(q_i) \Rightarrow T_i(r_i)$
            \EndFor
            \EndProcedure
    \end{algorithmic}
    \label{algorithm:comp}
\end{algorithm}

Algorithm~\ref{algorithm:comp} shows the procedure RuntimeComposition that corresponds to a one step product of the input TES transition systems. Note that such procedure applied recursively on its results would generate a behavior that is in behavior of the product of the TES transition systems.

\paragraph{Results.}
Initially, the system consists of three \emph{trolls}, with identifiers $\mathit{id(0)}$, $\mathit{id(1)}$, and $\mathit{id(2)}$, each coordinated by two protocols $\mathit{swap(id(i),id(j))}$ with $i,j \in \{0,1,2\}$ and $j<i$. The trolls move on a grid and trolls $\mathit{id(0)}$, $\mathit{id(1)}$, and $\mathit{id(2)}$ are respectively initialized at position $(2;0)$, $(1;0)$, and $(0;0)$. 
\footnote{We refer to~\cite{clavel-etal-07maudebook} for a more detailed description of the Maude framework.}
 The property $P_\sorted$ is a reachability property on the state of the grid, that states that \emph{eventually, all robots are in the sorted position}. In Maude, we express such reachability property with the following search command:
 \begin{lstlisting}
 search [1] init =>* 
    [sys::Sys  [ field : Field | k((0;0)) |-> d(id(0)), 
                k((1;0)) |-> d(id(1)), 
                k((2;0)) |-> d(id(2)) ; true ; null]] .
  \end{lstlisting}

\newcommand{\battery}{\mathit{battery-out}}
Table~\ref{table:results} features three variations on the sorting problem.
The first system is composed of robots whose move are free on the grid. 
The second adds one battery for each component, whose energy level decreases for each robot move.
The third system adds a swap protocol for every pair of two robots.
The last system adds protocol and batteries to compose with the robots.

We record, for each of those systems, whether the sorted configuration is reachable ($P_\sorted$), and if all three robots can run out of energy ($P_\battery$). 
\begin{table}
    \caption{Evaluation of different systems for several behavioral properties.}
    \label{table:results}
    \def\arraystretch{1.5}
    \setlength{\tabcolsep}{0.5em}
    \centering
    \begin{tabular}{l|c|c}
        System &$P_\sorted$& $P_b$     \\
        \hline
        $\underset{0\leq i< 2}{\bowtie} R_i \bowtie G$                               &$12.10^3$ states, 25s, $31.10^6$ rewrites& - \\
        $\underset{0\leq i< 2}{\bowtie} (R_i \bowtie B_i) \bowtie G$                               &$12.10^3$ states, 25s, $31.10^6$ rewrites& true \\
        $\underset{0\leq i< 2}{\bowtie} R_i \bowtie G \underset{0\leq i<j<2}{\bowtie} S(R_i,R_j)$      &$8250$ states, 44s, $80.10^6$ rewrites & - \\
        $\underset{0\leq i< 2}{\bowtie} (R_i \bowtie B_i) \bowtie G \underset{0\leq i<j<2}{\bowtie} S(R_i,R_j)$      &$8250$ states, 71s, $83.10^6$ rewrites & false \\
    \end{tabular}
\end{table}
Observe that the reachability query returns a solution for both system: the one with and without protocols.
However, the time to reach the first solution increases as the number of states and transition increases (adding the protocol components). We leave as future work some optimization to improve on our results.

\section{Related work}

\paragraph{Discrete Event Systems}
Our work represents both cyber and physical aspects of system with a unified model of discrete event systems.
In~\cite{doi:10.1146/annurev-control-053018-023659}, the author lists the current challenges in modelling cyber-physical in such a way. 
The author points to the problem of modular control, where even though two modules run without problems in isolation, the same two modules may block when they are used in conjunction.

In~\cite{DBLP:journals/tac/SampathLT98}, the authors present procedures to synthesize supervisors that control a set of interacting processes and, in the case of failure, report a diagnosis. 
Cyber-physical systems have also been studied from an actor-model perspective, where actors interact through
events~\cite{Talcott2008,DBLP:journals/scp/KappeLAT19}.
In our work, we add to the event structure a timing constraint, and expose condition to take the product of discrete event systems at runtime.


\paragraph{Synchronization of processes}
In~\cite{Nivat1982}, the author describes infinite behaviors of process and their synchronization. Notably, the problem of non-blockingness is stated: if two processes eventually interact on some actions, how to make sure that both processes will not block each others. The concept of centrality of a process is introduced, which corresponds to our definition of deadlock freedom. Our work extends this work in that it gives sufficient conditions to compute at runtime, using a step-wise product, some behaviors in the synchronous product.


\paragraph{Components}
In~\cite{AR02}, the authors give a co-inductive definition of components, to which \cite{DBLP:journals/corr/abs-2110-02214} is an extension. 
In~\cite{BAIER200675}, the authors propose a state based specification as constraint automata. A transition in a constraint automata is labeled by a guarded command, whose satisfaction depends on the context of its product (other constraint automata). Except from~\cite{KOKASH201311}, constraint automata do not have time as part of their semantics (i.e., only specify time insensitive components), and only describe observables on ports. In that respect, our model extends constraint automata by generalizing the set of possible observables, and adding the time of the observable as part of the transition.


%

\paragraph{Timed systems} 
In~\cite{DBLP:journals/tcs/FiadeiroL17}, the authors use
heterogeneous timed asynchronous relational nets (HT-ARNs) to model timed sensitive components, and a specification as timed IO-automata.
The authors show some conditions (progress-enabledness and r-closure) for the product of two HT-ARNs to preserve progress-enabledness.
We may have recover a similar result, but with some modifications. Our product is more expressive: $\kappa$ needs not be only synchronization of shared events, but can have more intricate coordination~\cite{DBLP:journals/corr/abs-2110-02214} (e.g., exclusion of two events).
We do not necessitate our process to be $r$-closed, and in general, we do not want to explicitly write the silent observations.

\section{Conclusion}
We introduced an operational specification of cyber-physical components introduced in~\cite{DBLP:journals/corr/abs-2110-02214}, for which we gave sufficient conditions for some products to be correctly implemented at runtime.
As a consequence, we implemented our framework in Maude, and use our correctness result to show that our analysis are sound.
We finally proved an emergent sorting property our of a set of robots whose moves on a grid are coordinated by a local swapping protocol.

\paragraph{Acknowledgement}
Talcott was partially supported by the U. S. Office of Naval Research under award numbers N00014-15-1-2202 and N00014-20-1-2644, and NRL grant N0017317-1-G002.
Arbab was partially supported by the U. S. Office of Naval Research under award number N00014-20-1-2644.

\printbibliography
\newpage
\section{Appendix}

We recall a list of results for showing algebraic properties of product on components given properties of composability relation on observables. 

\begin{property}
    We list a series of properties on a composability relation on observable $\sqcap$ where, for $O_1, O_2, O_3 \subseteq \E$:
    \begin{enumerate}
        \item\label{sq1}
            if $\ind_\sqcap(O_1, O_2)$ and $\ind_\sqcap(O_1, O_3)$ then $\ind_\sqcap(O_1, O_2 \cup O_3)$; and
            if $\ind_\sqcap(O_2, O_3)$ and $\ind_\sqcap(O_1, O_3)$ then $\ind_\sqcap(O_1\cup O_2, O_3)$
        \item\label{sq2} if $(O_1, O_2) \in \sqcap$ and $(O_2, O_3) \in \sqcap$, then $(O_1, O_3) \in \sqcap$;
        \item\label{sq3} if $(O_1, O_2) \in \sqcap$ and $(O_1, O_3) \in \sqcap$, then $(O_2, O_3) \in \sqcap$ and $(O_3, O_2) \in \sqcap$;
        \item\label{sq4} if $(O_1, O_3) \in \sqcap$ and $(O_2, O_3) \in \sqcap$, then $(O_1, O_2) \in \sqcap$ and $(O_2, O_1) \in \sqcap$;
        \item\label{sq5}
            if $(O_1, O_2) \in \sqcap$ and $(O_1',O_2') \in \sqcap$, then $(O_1\cup O_1', O_2 \cup O_2') \in \sqcap$.
    \end{enumerate}
    \label{prop:sqcap}
\end{property}
\begin{corollary}
    \label{corollary:lifted-sync}
    Let $\sqcap$ be a composability relation on observables and $\kappa^\ssync_\sqcap$ the synchronous composability relation on observations.
    Then:
    \begin{itemize}
        \item if $\sqcap$ is symmetric, then the product $\times_{([\kappa^\ssync_\sqcap],[\cup])}$ is commutative;
        \item if $\sqcap$ satisfies Properties~\ref{prop:sqcap}, then the product $\times_{([\kappa^\ssync_\sqcap],[\cup])}$ is associative;
        \item if $\sqcap$ is co-reflexive, i.e., $(O_1,O_2) \in \sqcap$ implies $O_1 = O_2$,  then the product $\times_{([\kappa^\ssync_\sqcap],[\cup])}$ is idempotent.
    \end{itemize}
\end{corollary}

\end{document}